\documentclass[11pt]{article}
\usepackage{fullpage}
\usepackage{epsfig}
\usepackage{amsmath,amssymb,amsthm}
\usepackage{enumerate, multirow}
\usepackage{color}
\usepackage{authblk}
\usepackage{array}

\newcolumntype{P}[1]{>{\centering\arraybackslash}p{#1}}

\newtheorem {thm}{Theorem}[section]
\newtheorem {lem}[thm]{Lemma}
\newtheorem {cor}[thm]{Corollary}

\newtheorem {cnst}{Construction}[section]

\newtheorem {defn}{Definition}[section]
\newtheorem {exam}{Example}[section]

\newcommand{\SHF}{\text{$\mathsf{SHF}$}}

\begin{document}

\title{A tight bound on the size of certain separating hash families}
\author{Chuan Guo\thanks{C.~Guo's research is supported by NSERC CGS-M scholarship.} }
\author{Douglas~R.~Stinson\thanks{D.~Stinson's research is supported by NSERC discovery grant 203114-11.}}
\affil{David R.\ Cheriton School of Computer Science\\University of Waterloo\\
Waterloo, Ontario, N2L 3G1, Canada}
\date{\today}

\maketitle

\begin{abstract}
In this paper, we present a new lower bound on the size of separating hash families of type $\{w_1^{q-1},w_2\}$ where $w_1 < w_2$. Our result extends the paper by Guo et al.\ on binary frameproof codes \cite{GST2015}. This bound compares  well against known general bounds, and is especially useful when trying to bound the size of strong separating hash families. We also show that our new bound is tight by constructing hash families that meet
the new bound with equality.
\end{abstract}

\section{Introduction}

Let $X,Y$ be finite sets of size $n$ and $q$, respectively. Let $\mathcal{F}$ be a family of functions from $X$ to $Y$ with $\mathcal{F} = N$. Given positive integers $w_1,w_2,\ldots,w_t$, we say that $\mathcal{F}$ is a $\{w_1,w_2,\ldots,w_t\}$-\emph{separating hash family}, denoted $\SHF(N;n,q,\{w_1,w_2,\ldots,w_t\})$, if for every choice of subsets $X_1,X_2,\ldots,X_t \subseteq X$ with $|X_i| = w_i$ for $i = 1,\ldots,t$ and $X_i \cap X_j = \emptyset$ for $i \neq j$, there exists some $f \in \mathcal{F}$ such that $f(X_i) \cap f(X_j) = \emptyset$ for $i \neq j$. Such $f$ is said to \emph{separate} the sets $X_1,\ldots,X_t$. The parameter multiset $\{w_1,w_2,\ldots,w_t\}$ is called the \emph{type} of the SHF.

The notion of separating hash families was introduced by Stinson et al.\ in \cite{SRC1998}. It is a generalization of many other combinatorial structures such as perfect hash families \cite{Mehlhorn1982}, frameproof codes \cite{Boneh1998}, and secure frameproof codes \cite{STW2000}. We would like to study bounds on the size of separating hash families when given the other parameters.

It is often useful to represent separating hash families in matrix form. When given an $\SHF(N;n,$\\$q,\{w_1,w_2,\ldots,w_t\})$, construct an $N \times n$ $q$-ary matrix $A$ with $A(i,j) = f_i(x_j)$ where $f_1,\ldots,f_N$ is some fixed ordering of the functions in $\mathcal{F}$ and $x_1,\ldots,x_n$ is some fixed ordering of the elements of $X$. This matrix is called the \emph{representation matrix} of $\mathcal{F}$. Specializing our definition of an SHF to this form, the equivalent property for when a matrix is the representation matrix of an SHF is as follows.

\begin{thm}
\label{t-equiv}
An $N \times n$ $q$-ary matrix $A$ is the representation matrix of an $\SHF(N;n,q,\{w_1,w_2,\ldots,$\\$w_t\})$ if and only if, for every choice of $t$ column sets $C_1,\ldots,C_t$ in $A$ where $C_i \cap C_j = \emptyset$ for $i \neq j$ and $|C_i| = w_i$ for $i = 1,\ldots,t$, there exists a row $r$ such that $M(r,c_i) \neq A(r,c_j)$ whenever $c_i \in C_i$ and $c_j \in C_j$ where $i \neq j$.
\end{thm}

A list of $t$ column sets $(C_1,\ldots,C_t)$, as specified in Theorem \ref{t-equiv}, 
will be termed a {\it column set $t$-tuple}.

We will only consider SHFs with $\sum_i w_i \leq n$ and $q \geq t$ in order to avoid vacuous cases. The following properties regarding SHFs with different parameter sets $\{w_1,\ldots,w_t\}$ are easy to prove.

\begin{thm}
\label{param-change}
Let $\mathcal{F}$ be an $\SHF(N;n,q,\{w_1,w_2,\ldots,w_t\})$ with $\sum_i w_i \leq n$ and $q \geq t$.
\begin{enumerate}[(i)]
\item If $w_1' \leq w_1$ then $\mathcal{F}$ is also an $\SHF(N;n,q,\{w_1',w_2,\ldots,w_t\})$.
\item If $w_1' = w_1 + w_2$ then $\mathcal{F}$ is also an $\SHF(N;n,q,\{w_1',w_3,\ldots,w_t\})$.
\end{enumerate}
\end{thm}

We now present some known results on general separating hash families.

\begin{thm}[\cite{BESZ2008}]
\label{BESZ2008}
If there exists an $\SHF(N;n,q,\{w_1,\ldots,w_t\})$ with $w_1,w_2 \leq w_i$ for $i = 3,\ldots,t$, then $$n \leq \gamma q^{\lceil \frac{N}{u-1} \rceil},$$ where $u = \sum_i w_i$ and $\gamma = (w_1 w_2 + u - w_1 - w_2)$.
\end{thm}

\begin{thm}[\cite{BazTran2011}]
\label{BazTran2011}
If there exists an $\SHF(N;n,q,\{w_1,\ldots,w_t\})$, then $$n \leq (u-1)q^{\lceil \frac{N}{u-1} \rceil},$$ where $u = \sum_i w_i$.
\end{thm}

\begin{thm}[\cite{BazTran2013}]
\label{BazTran2013}
If there exists an $\SHF(N;n,q,\{w_1,\ldots,w_t\})$ with $t \geq 3$ and $u = \sum_i w_i \geq 4$, then $$n \leq (u-1)q^{\lceil \frac{N}{u-1} \rceil}.$$
\end{thm}

In the remainder of this paper, we will present a construction and a new bound on the size of an SHF of the type $\{w_1^{q-1},w_2\}$. Using Theorem \ref{param-change}, one can extend this result to  bounds for more general types of SHF, such as strong separating hash families \cite{SarkarStinson2001}.

\section{A construction for SHF of type $\{w_1^{q-1},w_2\}$}

We first give a construction for SHF of type $\{w_1^{q-1},w_2\}$.

\begin{cnst}
\label{construction1}
Fix positive integers $n,q,w_1,w_2$ with $w_1 < w_2$ and $w_2 + (q-1)w_1 \leq n$. Let $$\mathcal{S} = \{(C_1,\ldots,C_{q-1}) : C_i \subseteq \{1,\ldots,n\} \text{ with } |C_i| = w_1 \text{ for all } i \text{ and } C_i \cap C_j = \emptyset \text{ if } i \neq j\},$$ and let $$\mathcal{T} = \{(C_1,\ldots,C_{q-1}) \in \mathcal{S} : c_1 < c_2 < \ldots < c_{q-1} \text{ where } c_i \text{ is the smallest element of } C_i\}.$$
Now for $(C_1,\ldots,C_{q-1}) \in \mathcal{T}$, let $r_{(C_1,\ldots,C_{q-1})}$ be the vector
$$r_{(C_1,\ldots,C_{q-1})}(i) =
  \begin{cases} 
      j    \hfill & \text{ if $i \in C_j$} \\
      0    \hfill & \text{ otherwise.} \\
  \end{cases}$$
Let $A$  be the matrix that contains all rows $r_{(C_1,\ldots,C_{q-1})}$ for every $(C_1,\ldots,C_{q-1}) \in \mathcal{T}$.
\end{cnst}

\begin{thm}
\label{tight}
The matrix $A$ from Construction \ref{construction1} is  an $\SHF(N;n,q,\{w_1^{q-1},w_2\})$ where $$N = \frac{1}{(q-1)!} {n \choose w_1} {n-w_1 \choose w_1} \cdots {n-(q-2)w_1 \choose w_1}.$$
\end{thm}
\begin{proof}
Let $C_0,\ldots,C_{q-1}$ be pairwise disjoint subsets of $\{1,\ldots,n\}$ such that $|C_0| = w_2$ and $|C_i| = w_1$ for $i = 1,\ldots,q-1$. By construction, there exists a unique permutation $\pi$ over $\{1,\ldots,q-1\}$ such that the $(q-1)$-tuple $(C_{\pi(1)},\ldots,C_{\pi(q-1)})$ is contained in $\mathcal{T}$. The column set $q$-tuple is separated by the row $r_{(C_{\pi(1)},\ldots,C_{\pi(q-1)})}$ in $A$. Thus $A$ is the representation matrix of an SHF of type $\{w_1^{q-1},w_2\}$.

Clearly $A$ has $n$ columns and $|\mathcal{T}|$ rows. For any $(C_1,\ldots,C_{q-1}) \in \mathcal{T}$, every permutation $\pi$ over $\{1,\ldots,q-1\}$ gives a unique element $(C_{\pi(1)},\ldots,C_{\pi(q-1)}) \in \mathcal{S}$. Since there are 
\[{n \choose w_1} {n-w_1 \choose w_1} \cdots {n-(q-2)w_1 \choose w_1}\] elements in $\mathcal{S}$, we have that $$|\mathcal{T}| = \frac{1}{(q-1)!} {n \choose w_1} {n-w_1 \choose w_1} \cdots {n-(q-2)w_1 \choose w_1},$$ as desired.
\end{proof}

\section{A bound for SHF of type $\{w_1^{q-1},w_2\}$}

In this section, for a certain range of values $n$, we prove a lower bound on $N$ 
for existence of an $\SHF(N;n,q,\{w_1^{q-1},w_2\})$, where $w_1^{q-1}$ denotes the multiset consisting of $q-1$ copies of $w_1$ and  $w_1 < w_2$. Whenever it is applicable, this lower bound is  tight, in view of Theorem \ref{tight}.

Our bound is in fact a generalization of Theorem 2.2.3 in \cite{GST2015}, which we provide here for reference.

\begin{thm}[\cite{GST2015}]
\label{GST}
Let $w,N$ be positive integers such that $w \geq 3$ and $w+1 \leq N \leq 2w+1$. Suppose there exists an $\SHF(N;n,2,\{1,w\})$. Then $n \leq N$.
\end{thm}

We will extend the idea of the proof of Theorem 2.2.3 in \cite{GST2015} by counting the total number of column set $q$-tuples separated in an SHF versus the number of column set $q$-tuples separated by a single row in the SHF. We can then give a lower bound on the number of rows required by dividing these two quantities. The following definition will be used throughout this section.

\begin{defn}
Let $x \in Q^n$ where $Q = \{0,1,\ldots,q-1\}$. We say that $x$ is of \emph{weight} $(i_1,i_2,\ldots,i_{q-1})$ if the number of entries equal to $k$ in $x$ is exactly $i_k$, for each $k = 1,\ldots,q-1$. The number of entries equal to 0 is thus $i_0 = n - \sum_{k=1}^{q-1} i_k$.
\end{defn}

The next definition gives a simplified notation for counting the number of column set $q$-tuples separated by a row of weight $(i_1,i_2,\ldots,i_{q-1})$. The correctness of this fact will be proven in Lemma \ref{num-separated}.

\begin{defn}
Let $w_1,w_2$ be positive integers with $w_1 < w_2$. For integers $i_0,i_1,\ldots,i_{q-1}$ with $i_0 \geq w_2$, $i_k \geq w_1$ for $k = 1,\ldots,q-1$ and $n \geq \sum_{k=0}^{q-1} i_k$, define $$T_{w_1,w_2,n}^{(q-1)}(i_1,\ldots,i_{q-1}) = {i_1 \choose w_1} {i_2 \choose w_1} \cdots {i_{q-1} \choose w_1} {n - \sum_{k=1}^{q-1} i_k \choose w_2}.$$
\end{defn}

\begin{lem}
\label{num-separated}
Let $w_1,w_2$ be positive integers with $w_1 < w_2$. For integers $i_0,i_1,\ldots,i_{q-1}$ with $i_0 \geq w_2$, $w_1 \leq i_k < w_2$ for $k = 1,\ldots,q-1$ and $n \geq \sum_{k=0}^{q-1} i_k$, the number of column set $q$-tuples separated by a row of weight $(i_1,\ldots,i_{q-1})$ is $$Z = (q-1)! \; T_{w_1,w_2,n}^{(q-1)}(i_1,\ldots,i_{q-1}).$$
\end{lem}
\begin{proof}
Since $w_1 \leq i_k < w_2$ for $k = 1,\ldots,q-1$, it is clear that a row $r$ of weight $(i_1,\ldots,i_{q-1})$ only separates column set $q$-tuples of the form $(C_1,\ldots,C_q)$ with $|C_k| = w_1$ for $k = 1,\ldots,q-1$ and $|C_q| = w_2$. The columns in $C_q$ correspond to entries in $r$ that are equal to 0. The columns in $C_k$ for $k = 1,\ldots,q-1$ correspond to distinct entries in $r$ that are equal to $1,\ldots,q-1$. There are $(q-1)!$ permutations of the set $\{1,\ldots,q-1\}$, thus the total number of columns set $q$-tuples separated by $r$ is
\begin{align*}
Z &= (q-1)! {i_0 \choose w_2} {i_1 \choose w_1} {i_2 \choose w_1} \cdots {i_{q-1} \choose w_1} \\
  &= (q-1)! \;   T_{w_1,w_2,n}^{(q-1)}(i_1,\ldots,i_{q-1}).
\end{align*}
\end{proof}

Using Lemma \ref{num-separated}, we would like to determine the maximum number of column set $q$-tuples separated by a row of weight $(i_1,\ldots,i_{q-1})$. The following lemma shows that this maximum is achieved when 
$i_1 = \cdots = i_{q-1} = w_1$.

\begin{lem}
\label{global-max}
Let $w_1,w_2$ be positive integers such that $w_1 < w_2$, and let $q,n$ be positive integers with $q \geq 2$ and $$w_2 + (q-1)w_1 \leq n \leq w_2 + (q-1)w_1 + \frac{w_2}{w_1} - 1.$$ Then for every $k = 1,\ldots,q-1$, we have $$T_{w_1,w_2,n}^{(q-1)}(i_1,\ldots,i_{q-1}) > T_{w_1,w_2,n}^{(q-1)}(i_1,\ldots,i_{k-1},i_k+1,i_{k+1},\ldots,i_{q-1}).$$ In particular, $T_{w_1,w_2,n}^{(q-1)}$ obtains its global minimum at $(w_1,\ldots,w_1)$ over the domain of integers $(i_1,\ldots,i_{q-1})$ for which $T_{w_1,w_2,n}^{(q-1)}$ is defined.
\end{lem}
\begin{proof}
\begin{align*}
&\hspace{20pt} T_{w_1,w_2,n}^{(q-1)}(i_1,\ldots,i_{q-1}) > T_{w_1,w_2,n}^{(q-1)}(i_1,\ldots,i_{k-1},i_k+1,i_{k+1},\ldots,i_{q-1}) \\
&\Leftrightarrow {i_k \choose w_1} {n - \sum_{l=1}^{q-1} i_l \choose w_2} > {i_k + 1 \choose w_1} {n - \sum_{l=1}^{q-1} i_l - 1 \choose w_2} \\
&\Leftrightarrow \frac{i_k - w_1 + 1}{i_k + 1} > \frac{n - \sum_{l=1}^{q-1} i_l - w_2}{n - \sum_{l=1}^{q-1} i_l}
\end{align*}
Letting $I = \sum_{l=1}^{q-1} i_l$ and rearranging the inequality gives
\begin{align*}
&\hspace{20pt} (i_k + 1 - w_1)(n - I) > (n - I - w_2)(i_k + 1) \\
&\Leftrightarrow -w_1(n - I) > -w_2(i_k + 1) \\
&\Leftrightarrow n\frac{w_1}{w_2} < i_k + 1 + \frac{w_1}{w_2}I \\
&\Leftrightarrow n < i_k \frac{w_2}{w_1} + I + \frac{w_2}{w_1}
\end{align*}
where the last inequality holds by the assumption $n < w_2 + (q-1)w_1 + \frac{w_2}{w_1}$ since $w_1 \leq i_k$ and $(q-1)w_1 \leq I$.
\end{proof}

Before we prove the main theorem, we need a final lemma that corresponds to a special case.

\begin{lem}
\label{special-case}
Let $q,w$ be positive integers with $q \geq 3$ and $w \geq 2$. Let $n = 2w + q - 2$. Then $$(q-1)! \;   T_{1,w,n}^{(q-1)}(1,\ldots,1) > 2(q-2)! \;   T_{1,w,n}^{(q-1)}(1,\ldots,1,w).$$
\end{lem}
\begin{proof}
Expanding the desired inequality gives
\begin{align*}
&\hspace{20pt} (q-1)!  {1 \choose 1}^{q-1} {n-q+1 \choose w} > 2(q-2)!  {1 \choose 1}^{q-2} {w \choose 1} {w \choose w} \\
&\Leftrightarrow (q-1){2w-1 \choose w} > 2w.
\end{align*}
One can check that ${2w-1 \choose w} > w$ for $w \geq 2$, and the proof follows since $q-1 \geq 2$.
\end{proof}

\begin{thm}
\label{main-theorem}
Let $w_1,w_2$ be positive integers with $w_1 < w_2$, and let $q,n$ be positive integers with 
$q \geq 2$ and 
\begin{equation}
\label{eq.thm2.5}
w_2 + (q-1)w_1 \leq n \leq w_2 + (q-1)w_1 + \frac{w_2}{w_1} - 1.
\end{equation} 
If there exists an $\SHF(N;n,q,\{w_1^{q-1},w_2\})$ then $$N \geq \frac{1}{(q-1)!} {n \choose w_1} {n-w_1 \choose w_1} \cdots {n - (q-2)w_1 \choose w_1}.$$
\end{thm}
\begin{proof}
Let $A$ be the representation matrix of an $\SHF(N;n,q,\{w_1^{q-1},w_2\})$. For any row $r$ of $A$ and $k \in \{0,1,\ldots,q-1\}$, let $i_k$ be the number of occurrences of symbol $k$ in row $r$. By permuting the alphabet on row $r$ if necessary, we may assume without loss of generality that $i_1 \leq i_2 \leq \ldots \leq i_{q-1} \leq i_0$. Furthermore, we may assume that $i_1 \geq w_1$ and $i_0 \geq w_2$, since otherwise $r$ cannot separate any column set $q$-tuple $(C_0,C_1,\ldots,C_{q-1})$ with $|C_k| = w_1$ for $1 \leq k \leq q-1$ and $|C_0| = w_2$ and we may remove $r$ from the matrix.
Observe that
\begin{align*}
i_{q-1} &= n - i_0 - \sum_{k=1}^{q-2} i_k \\
				&\leq n - w_2 - (q-2)w_1 \\
				&\leq w_1 + \frac{w_2}{w_1} - 1  \qquad\qquad \text{from (\ref{eq.thm2.5})}\\
				&\leq w_1 + (w_2 - w_1) \\
				&= w_2.
\end{align*}
We consider the following two cases.
\begin{enumerate}[(i)]
\item $i_{q-1} = w_2$. The above inequalities must all be equalities, so we have $w_1 = 1$, $i_k = 1$ for $k = 1,\ldots,q-2$, $i_0 = w_2$ and $$n = w_2 + (q-1)w_1 + \frac{w_2}{w_1} - 1 = 2w_2 + q - 2.$$  Let $w = w_2$. We only need to consider the case $q \geq 3$ since $q=2$ is covered by Theorem \ref{GST}. The number of column set $q$-tuples separated by $r$ is exactly $2(q-2)! \;   T_{1,w,n}^{(q-1)}(1,\ldots,1,w)$, which is less than the number of column set $q$-tuples separated by a row of weight $(w_1,\ldots,w_1) = (1,\ldots,1)$ by Lemma \ref{num-separated} and Lemma \ref{special-case}.
\item $i_{q-1} < w_2$: By Lemma \ref{num-separated}, the number of column set $q$-tuples separated by $r$ is $$Z = (q-1)! \;   T_{w_1,w_2,n}^{(q-1)}(i_1,\ldots,i_{q-1}).$$ The number of column set $q$-tuples separated by a row of weight $(w_1,\ldots,w_1)$ is greater than $Z$ by Lemma \ref{global-max} unless $i_k = w_1$ for $k = 1,\ldots,q-1$.
\end{enumerate}
In either case, the number of column set $q$-tuples separated by $r$ is maximized only when the row is of weight $(w_1,\ldots,w_1)$. The total number of column set $q$-tuples that need to be separated is $$T = {n \choose w_1} {n-w_1 \choose w_1} \cdots {n-(q-1)w_1 \choose w_2}.$$ Thus
\begin{align*}
N &\geq \frac{T}{(q-1)! \;   T_{w_1,w_2,n}^{(q-1)}(w_1,\ldots,w_1)} \\
  &= \frac{1}{(q-1)!} {n \choose w_1} {n-w_1 \choose w_1} \cdots {n-(q-2)w_1 \choose w_1}.
\end{align*}
\end{proof}

The following result is an immediate consequence of Theorems \ref{tight} and \ref{main-theorem}.
\begin{cor}
\label{exact}
Let $w_1,w_2$ be positive integers with $w_1 < w_2$, and let $q,n$ be positive integers with 
$q \geq 2$ and 
\[
w_2 + (q-1)w_1 \leq n \leq w_2 + (q-1)w_1 + \frac{w_2}{w_1} - 1.
\]
Then the minimum value of $N$ such that there exists  
an $\SHF(N;n,q,\{w_1^{q-1},w_2\})$ is $$N = \frac{1}{(q-1)!} {n \choose w_1} {n-w_1 \choose w_1} \cdots {n - (q-2)w_1 \choose w_1}.$$
\end{cor}

\section{Applications}

Theorem \ref{main-theorem} is particularly useful for studying the combinatorial objects known as strong separating hash families (denoted SSHF), introduced by Sarkar and Stinson in \cite{SarkarStinson2001}. They are equivalent to an SHF of type $\{1^{t_1},t_2\}$ for some positive integers $t_1,t_2$. We can give a strong bound for the code length of SSHFs as a corollary.

\begin{cor}
\label{sshf}
Let $n,t_1,t_2$ be positive integers with $t_1 \geq q-1$ and $t_1 + t_2 \leq n \leq 2(t_1 + t_2) - q$. Suppose there exists an $\SHF(N;n,q,\{1^{t_1},t_2\})$. Then $$N \geq {n \choose q-1}.$$
\end{cor}
\begin{proof}
By Theorem \ref{param-change}, an $\SHF(N;n,q,\{1^{t_1},t_2\})$ is also an $\SHF(N;n,q,\{1^{q-1},t_1 + t_2 - q + 1\})$. Applying Theorem \ref{main-theorem}, if $t_1 + t_2 \leq n \leq 2(t_1+t_2)-q$, then we have $$N \geq \frac{1}{(q-1)!}n(n-1)\ldots,(n-q+2),$$ as desired.
\end{proof}

\begin{exam}
Let $q = 3$, $t_1 = 4$ and $t_2 = 3$. Suppose there exists an $\SHF(N;11,3,\{1,1,1,1,3\})$ (Corollary \ref{sshf} applies to $n = 7,8,9,10$ as well). Then $N \geq {11 \choose 2} = 55$. In other words, for $N \leq 54$, we have that $n \leq 10$.

Compare this with known results: Theorem \ref{BESZ2008} and Theorem \ref{BazTran2011} both give the bound $n \leq 6(3^9) = 118098$ for $N = 54$; Theorem \ref{BazTran2013} gives the bound $$n \leq 6(3^9) + 2 - 2 \sqrt{3 (3^9) + 1} < 118023$$ for $N = 54$.
\end{exam}

Finally, Table \ref{computed-values} lists various parameter choices for $q,w_1,w_2$ and compares the bound in Theorem \ref{main-theorem} to some known bounds for general SHFs. The symbol $\Omega$ means the computed bound is above the Java $\texttt{double}$ maximum value of $(2 - 2^{-52}) 2^{1023}$.

\begin{table}[p]
\centering
\label{computed-values}
\begin{tabular}{|P{12pt}|P{12pt}|P{12pt}|c|c|c|c|c|}
  \hline
  \multirow{2}{*}{$q$} & \multirow{2}{*}{$w_1$} & \multirow{2}{*}{$w_2$} & \multirow{2}{*}{$N \leq$} & \multicolumn{4}{c|}{implies $n \leq$} \\
	\cline{5-8}
	& & & & Theorem \ref{main-theorem} & Theorem \ref{BESZ2008} & Theorem \ref{BazTran2011} & Theorem \ref{BazTran2013} \\
	\hline
	\hline
	3 & 1 & 2 & 9 & 4 & 243 & 243 & 213 \\
	\hline
	3 & 1 & 3 & 20 & 6 & 2916 & 2916 & 2824 \\
	\hline
	3 & 1 & 4 & 35 & 8 & 32805 & 32805 & 32526 \\
	\hline
	3 & 1 & 5 & 54 & 10 & 354294 & 354294 & 353454 \\
	\hline
	3 & 1 & 6 & 77 & 12 & 3720087 & 3720087 & 3717563 \\
	\hline
	3 & 2 & 3 & 104 & 6 & $3.09^9$ & $2.32^9$ & $2.32^9$ \\
	\hline
	3 & 2 & 4 & 377 & 8 & $5.81^{26}$ & $4.07^{26}$ & $4.07^{26}$ \\
	\hline
	3 & 2 & 5 & 629 & 9 & $5.91^{38}$ & $3.94^{38}$ & $3.94^{38}$ \\
	\hline
	3 & 2 & 6 & 1484 & 11 & $7.43^{79}$ & $4.77^{79}$ & $4.77^{79}$ \\
	\hline
	3 & 3 & 4 & 2099 & 9 & $6.64^{112}$ & $3.98^{112}$ & $3.98^{112}$ \\
	\hline
	3 & 3 & 5 & 4619 & 10 & $4.84^{221}$ & $2.69^{221}$ & $2.69^{221}$ \\
	\hline
	3 & 3 & 6 & 17159 & 12 & $\Omega$ & $\Omega$ & $\Omega$ \\
	\hline
	4 & 1 & 2 & 19 & 5 & 4096 & 4096 & 3987 \\
	\hline
	4 & 1 & 3 & 54 & 7 & $2.09^7$ & $2.09^7$ & $2.09^7$ \\
	\hline
	4 & 1 & 4 & 118 & 9 & $6.59^{12}$ & $6.59^{12}$ & $6.59^{12}$ \\
	\hline
	4 & 1 & 5 & 219 & 11 & $1.29^{20}$ & $1.29^{20}$ & $1.29^{20}$ \\
	\hline
	4 & 1 & 6 & 362 & 13 & $3.96^{28}$ & $3.96^{28}$ & $3.96^{28}$ \\
	\hline
	4 & 2 & 3 & 1259 & 8 & $1.33^{96}$ & $1.06^{96}$ & $1.06^{96}$ \\
	\hline
	4 & 2 & 4 & 6929 & 10 & $\Omega$ & $\Omega$ & $\Omega$ \\
	\hline
	4 & 2 & 5 & 13859 & 11 & $\Omega$ & $\Omega$ & $\Omega$ \\
	\hline
	4 & 2 & 6 & 45044 & 13 & $\Omega$ & $\Omega$ & $\Omega$ \\
	\hline
	4 & 3 & 4 & 200199 & 12 & $\Omega$ & $\Omega$ & $\Omega$ \\
	\hline
	4 & 3 & 5 & 560559 & 13 & $\Omega$ & $\Omega$ & $\Omega$ \\
	\hline
	4 & 3 & 6 & 3203199 & 15 & $\Omega$ & $\Omega$ & $\Omega$ \\
	\hline
	5 & 1 & 2 & 33 & 6 & 390625 & 390625 & 389658 \\
	\hline
	5 & 1 & 3 & 125 & 8 & $2.86^{15}$ & $2.86^{15}$ & $2.86^{15}$ \\
	\hline
	5 & 1 & 4 & 329 & 10 & $2.48^{34}$ & $2.48^{34}$ & $2.48^{34}$ \\
	\hline
	5 & 1 & 5 & 714 & 12 & $6.46^{63}$ & $6.46^{63}$ & $6.46^{63}$ \\
	\hline
	5 & 1 & 6 & 1364 & 14 & $1.57^{107}$ & $1.57^{107}$ & $1.57^{107}$ \\
	\hline
	5 & 2 & 3 & 17324 & 10 & $\Omega$ & $\Omega$ & $\Omega$ \\
	\hline
	5 & 2 & 4 & 135134 & 12 & $\Omega$ & $\Omega$ & $\Omega$ \\
	\hline
	5 & 2 & 5 & 315314 & 13 & $\Omega$ & $\Omega$ & $\Omega$ \\
	\hline
	5 & 2 & 6 & 1351349 & 15 & $\Omega$ & $\Omega$ & $\Omega$ \\
	\hline
	5 & 3 & 4 & 28027999 & 15 & $\Omega$ & $\Omega$ & $\Omega$ \\
	\hline
	5 & 3 & 5 & 95295198 & 16 & $\Omega$ & $\Omega$ & $\Omega$ \\
	\hline
	5 & 3 & 6 & 775975199 & 18 & $\Omega$ & $\Omega$ & $\Omega$ \\
	\hline
\end{tabular}
\caption{Comparison of Bounds for $\SHF(N;n,q,\{w_1^{q-1},w_2\})$}
\end{table}

\section{Conclusion}

We have presented a new bound in Theorem \ref{main-theorem} for SHF of type $\{w_1^{q-1},w_2\}$. As an application, we derived a bound in Corollary \ref{sshf} for SSHFs that compares  well against known general bounds. One can also choose other types of SHFs and apply Theorem \ref{main-theorem} to obtain competitive bounds, since Table \ref{computed-values} demonstrates a large gap between our result and best known general bounds.

There is an inherent difficulty of generalizing Theorem \ref{main-theorem} to other types. For example, if we relax the type of the SHF to $\{w_1^{q-2},w_2,w_3\}$ where $w_1 < w_2 < w_3$, then a row of weight $(w_1,\ldots,w_1,w_2,w_2)$ could separate the column set consisting of $w_2$ columns in multiple ways. This difficulty is even more prevalent when the type set $\{w_1,\ldots,w_t\}$ consists of a large number of different values. It would be interesting to develop a counting method that can overcome this difficulty. Another extension of our result could be in the direction of allowing the type multiset $\{w_1,\ldots,w_t\}$ to contain more elements than $q$, i.e., $t > q$. Making progress in either direction would allow us to derive more powerful bounds for general SHFs.



\begin{thebibliography}{99}

\bibitem{BazTran2011} M. Bazrafshan and Tran van Trung.
                      Bounds for separating hash families.
                 {\it J.\ Combin.\ Theory A} {\bf 118} (2011), 1129--1135.

\bibitem{BazTran2013} M.~Bazrafshan and Tran van Trung.
											Improved bounds for separating hash families.
											{\it Designs, Codes, and Cryptography} (2013), 369--382.
 
\bibitem{BESZ2008} S.~R.~Blackburn, T.~Etzion, D.~R.~Stinson and G.~M.~Zaverucha.
									 A bound on the size of separating hash families.
									 {\it J.\ Combin.\ Theory Ser. A} {\bf 115} (2008), 1246--1256.

\bibitem{Boneh1998} D.~Boneh and J.~Shaw. Collusion-free fingerprinting
   for digital data, {\it IEEE Trans.\ Inform.\ Theory} {\bf 44} (1998),
   1897--1905.
				
\bibitem{GST2015} C.~Guo, D.~R.~Stinson and Tran van Trung.
				On tight bounds for binary frameproof codes.
				To appear in {\it Designs, Codes, and Cryptography}.
				
\bibitem{Mehlhorn1982} K.~Mehlhorn.
				On the program size of perfect and universal hash functions.
				{\it Proceedings of the 23rd Annual Symposium on Foundations of Computer Science} (1982), 170--175.

\bibitem{SarkarStinson2001} P.~ Sarkar and D.~R.~Stinson.
          Frameproof and IPP codes,
	  Progress in Cryptology -- Indocrypt 2001, 
         {\it Lecture Notes in Computer Science}, Springer,
	  {\bf 2247} (2001), 117--126.

\bibitem{STW2000} D.~R.~Stinson, Tran van Trung and R.~Wei.
    Secure frameproof codes, key distribution patterns, group
    testing algorithms and related structures,
    {\it J. Statist.\ Plann.\ Inference} {\bf 86} (2000), 595--617.
			
\bibitem{SRC1998} D.~R.~Stinson, R.~Wei and K.~Chen.
			 On generalized separating hash families.
			 {\it J. Combin. Theory Ser. A} {\bf 115} (2008), 105--120.


\end{thebibliography}
\end{document}